\documentclass[conference]{IEEEtran}
\IEEEoverridecommandlockouts
\usepackage{cite}
\usepackage{float}
\usepackage{amsmath,amssymb,amsfonts}
\usepackage{amsthm}
\usepackage{algorithmic}
\usepackage{graphicx}
\usepackage{subfigure}
\usepackage{textcomp}
\usepackage{xcolor}
\usepackage{array}
\usepackage{makecell}
\usepackage{graphicx}
\usepackage{multirow}
\usepackage{enumitem}
\usepackage{comment}
\usepackage{threeparttable}
\usepackage{lipsum}
\usepackage[hidelinks]{hyperref}
\usepackage{dblfloatfix}
\usepackage{cuted} 
\usepackage{caption} 
\usepackage{multirow}
\usepackage{booktabs}
\def\BibTeX{{\rm B\kern-.05em{\sc i\kern-.025em b}\kern-.08em
		T\kern-.1667em\lower.7ex\hbox{E}\kern-.125emX}}

\newtheorem{theorem}{\bf \emph{Theorem}}
\newtheorem{lemma}[theorem]{\bf \emph{Lemma}}

\newtheorem{proposition}[theorem]{Proposition}

\newtheorem{example}{Example}

\begin{document}

	\title{State-Derivative Feedback Control for Damping Low-Frequency Oscillations in Bulk Power Systems
    \\
		\thanks{\color{black} This work is funded in part by the National Science Foundation (NSF) grant ECCS-2145063.
        }
	}
	
\author{
\IEEEauthorblockN{
MST Rumi Akter$^{1}$, Anamitra Pal$^{2}$, and Rajasekhar Anguluri$^{1}$
}
\IEEEauthorblockA{
$^{1}$\textit{Department of Computer Science and Electrical Engineering, UMBC, Baltimore, MD 21250, USA}\\
$^{2}$\textit{School of Electrical, Computer, and Energy Engineering, ASU, Tempe, AZ 85281, USA}\\
Emails: \texttt{makter2@umbc.edu}, \texttt{anamitra.pal@asu.edu},
\texttt{rajangul@umbc.edu}
}
}

	\maketitle

\begin{abstract}
Low-frequency
oscillations remain a major challenge in 
bulk power systems with
high renewable penetration, long lines, and 
large loads.
Existing damping strategies based on power modulation of
high voltage DC (HVDC) or 
energy storage, are often limited by fixed control architectures, leaving some modes poorly damped. This paper introduces a state-derivative feedback (SDF) damping controller that uses both frequency and its rate of change as feedback signals. Incorporating state derivatives enhances modal damping and accelerates frequency recovery, enabling HVDC and energy storage
to effectively stabilize the grid. We evaluate the SDF controller on two- and three-area systems and compare performance with a
frequency
difference-based damping scheme. 
{\color{black} Results show that the SDF control reproduces state-feedback performance while providing good damping of both inter-  and 
intra-area
oscillations compared to the frequency-difference method, highlighting its potential as a practical solution for stabilizing power-electronics-rich grids.}
\end{abstract}


\section{Introduction}
\label{Intro}
The problem of damping unwanted low-frequency
oscillations (LFOs) in large power systems remains both theoretically intriguing and practically challenging. Over several decades, researchers have proposed a range of solutions, spanning from advanced control algorithms to state-of-the-art hardware implementations \cite{pal2013applying,pierre2019design}. Yet the issue persists, driven by the evolving dynamics of grids transitioning from conventional synchronous generation to hybrid architectures dominated by renewable sources and large loads. Thus, there remains a growing need for new 
controllers capable of effectively damping oscillations triggered by unexpected faults or load variations. 

The present work proposes a feedback control strategy that leverages 
area frequencies and
their derivatives (rate-of-change of frequency, RoCoF) to enhance system-wide oscillation damping. The motivation for using derivatives of states (rotor speeds and accelerations) instead of the states (rotor angles and speeds) is 
twofold: (i) the states 
need to be estimated
through \textit{nonlinear dynamic state estimation} methods~\cite{chai2025robust}; (ii) the derivatives of the states are related to phasor measurement unit (PMU) measurements via the (\textit{linear}) \textit{frequency divider formula} (FDF) \cite{milano2017frequency,sharma2025practical}. Because linear problems are solved more quickly
and accurately than equivalent nonlinear problems, it is more effective
to use state derivatives (instead of the states) as feedback for damping oscillations.

From an implementation standpoint, the proposed control can be realized through power modulation via actuators such as high voltage DC (HVDC) links or energy storage systems (ESSs), both of which have been previously employed
for LFO
damping~\cite{pal2013applying,neely2013benefits}.
However, instead of using the measured frequency differences as in \cite{neely2013benefits}, our work makes these actuators either inject or absorb active power based on both frequency and RoCoF signals. 
Finally, we
demonstrate the effectiveness of the proposed control law on benchmark power system models through detailed simulations under diverse system conditions (faults, data center load variations).
The key contributions of this work are as follows:
\begin{enumerate}
\item For an $N$-area system, 
a state-derivative feedback (SDF) control law is developed 
that yields the same closed-loop response as the state feedback (SF) control.
\item A detailed physical justification is provided for the assumptions required in SDF design, including the nonsingularity of the state matrix. A new result is presented linking this property to the synchronizing matrix, along with strategies to ensure nonsingularity.
\item The ability
of the proposed control to damp LFOs is
demonstrated on two- and three-area systems in presence of measurement errors, faults, and large load variations.
\end{enumerate}
	
	\smallskip
\textit{Related work}: {\color{black} 
Although derivative-based feedback systems have been developed for power systems (e.g., \cite{rakhshani2016inertia}), the focus has often been on single-loop designs with ad hoc gain tuning. Similarly, although prior works on LFO
damping have employed state-space methods to damp oscillations across multiple areas and operating conditions \cite{pal2013applying, neely2013benefits, Xu2021direct}, 
there has been 
very little work on using derivative signals for feedback control of oscillations \cite{liu2020model}.

In mechanical and structural systems, optimal-control-based SDF methods have been extensively studied for vibration and noise suppression \cite{abdelaziz2005state,da2012less,cardim2007design}.
Under mild assumptions, these studies show that SDF and conventional SF controllers share the same design structure and produce identical closed-loop responses. Building on this insight and design procedure, we demonstrate how frequency and RoCoF signals can be used for effective oscillation damping in power systems.}

\section{Problem Setup and Preliminaries}\label{sec: prelims}

In accordance with the logic developed in \cite{neely2013benefits},
we split the power system into multiple areas connected by tie-lines, with each area represented
by aggregated inertia and damping coefficients
(see Fig.~\ref{fig: two-area-motivation}). For an $N$-area system, let $N \times N$ diagonal matrices $M$ and $D$, respectively, contain the inertia and damping coefficients; and $T$ be the matrix of synchronizing torque coefficients. Then, the classical small-signal dynamics of the power system can
be expressed as:
\begin{equation}\label{eq: MMPS}
\begin{bmatrix}
    \Delta \dot{\delta} \\
\Delta \dot{\omega}
\end{bmatrix}\!=\!
\begin{bmatrix}
    0 & I \\
-M^{-1} T & -M^{-1} D
\end{bmatrix}
\begin{bmatrix}
    \Delta \delta \\
\Delta \omega
\end{bmatrix}
\!+\!\begin{bmatrix}
    0\\
    M^{-1}
\end{bmatrix}(u+\Delta P), 
\end{equation}
where the $N$-vectors $\Delta \delta$, $\Delta \omega$, and $\Delta P$ represent the deviations in relative rotor angles, speeds, and power balance (e.g., load mismatch or faults), respectively\footnote{The inertia constants have units of seconds (s), the relative rotor angles $\delta$ are in radians (rad), and the other quantities are in per unit (p.u.). 
}. Area frequency and speed deviations are treated equivalently under p.u. and small-signal assumptions.

The 
control input $u$ 
represents the active power injected or absorbed\footnote{In practice, this modulation can be achieved through HVDC links, fast-responding ESSs, or their combination.} at the tie-line buses. 
In its simplest form, the control input is chosen proportional to the difference in area frequencies; so for the two-area case shown in Fig.~\ref{fig: two-area-motivation}: 
\begin{align}\label{eq: simple_control_two_area_system}
    \begin{bmatrix}
        u_1\\
        u_2
    \end{bmatrix}&=-K_d\begin{bmatrix}
        1 & -1\\
        -1 & 1
    \end{bmatrix}\begin{bmatrix}
        \Delta \omega_1\\ 
        \Delta \omega_2
    \end{bmatrix}. 
\end{align}

Note that $u_1\!=\!-K_d(\Delta \omega_1 - \Delta \omega_2)$ with $u_2 = -u_1$. The gain $K_d$ is selected based on design considerations---the damping control block in Fig.~\ref{fig: two-area-motivation} implements this scheme. Three-area extensions of this scheme, representing the WECC model, are detailed in \cite{neely2013benefits}.
Note that this
analysis is generic and holds for arbitrary $N$; however, for illustration and simulation purposes, these smaller systems are used.

Although simple, the damping control input based solely on frequency differences \cite{neely2013benefits} is restrictive for several reasons:
\begin{enumerate}
\item It focuses on minimizing inter-area oscillations, while neglecting damping 
of intra-area frequencies and angles.
\item It enforces uniform damping for both $u_1$ and $u_2$.
\item It fails to accommodate additional feedback signals, such as RoCoF, that are readily available from PMUs.
\end{enumerate}

To overcome these limitations, we propose an SDF control strategy that utilizes area frequencies and RoCoFs to damp both absolute and differential oscillations in system states, not limited to frequencies alone.

Intuitively, one might expect that incorporating additional feedback signals offers greater flexibility for control. However, for the specific system considered in \eqref{eq: MMPS}, it is not immediately clear why frequency (or speed-deviation) and RoCoF feedback signals are compatible, given that the state variables of \eqref{eq: MMPS} are angles and frequencies. We address this dilemma by first studying the SDF control in a generic state-space form, and then relating it to the power system model.

  	\begin{figure}
		\centering
		\includegraphics[width=0.95\linewidth]{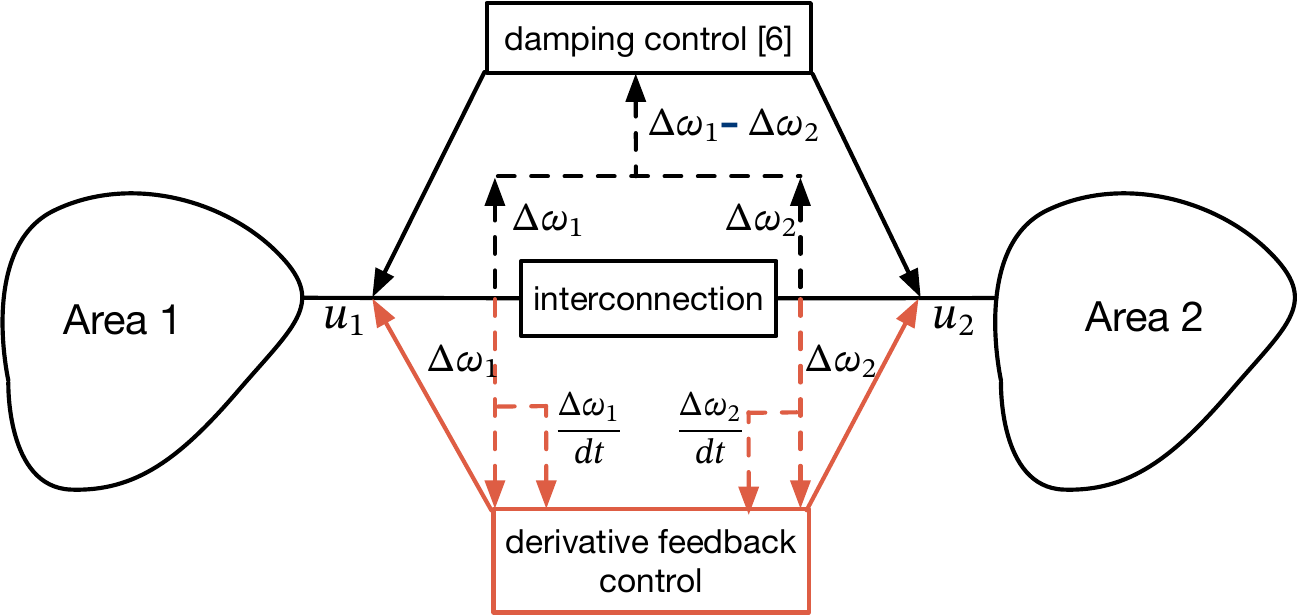}
		\caption{\small {\bf  Damping controller in \cite{neely2013benefits} vs. the proposed SDF control:} 
Two power-system areas are connected by a tie-line interconnection, which may be an HVDC transmission link or equipped with an energy storage device. The damping controller in \cite{neely2013benefits} regulates the source or sink power ($u$) in each area proportional to the frequency difference $\Delta\omega_1 - \Delta\omega_2$. In contrast, the SDF control regulates this power based on the individual frequencies $(\Delta\omega_1, \Delta\omega_2)$ and their derivatives $(\Delta\dot{\omega}_1, \Delta\dot{\omega}_2)$, providing greater flexibility through derivative terms and without restricting the control actuation via difference quantities.}
\label{fig: two-area-motivation}
        \end{figure}
\setlength{\textfloatsep}{2pt plus 1.0pt minus 2.0pt} 

\subsection{Primer on State-Derivative Feedback Control}\label{sec: primer}
Consider the controllable linear time-invariant dynamical system:
\begin{equation}
    \dot{x} = A x + B u,
    \label{eq: LTI}
\end{equation}
where the state $x \in \mathbb{R}^n$ and the control $u \in \mathbb{R}^m$.
The SF control law is given by:
\begin{equation}
    u_{\mathrm{SF}} = -K_s x + N_s r,
    \label{eq: SF}
\end{equation}
where $r \in \mathbb{R}^m$ is the reference, and matrices 
$K_s\in \mathbb{R}^{m\times n}$ and $N_s \in \mathbb{R}^{m\times m}$ are selected to achieve the desired closed-loop performance. Analogously, define the SDF control law:
\begin{equation}
    u_{\mathrm{SDF}} = -K_n \dot{x} + N_n r.
    \label{eq: SDF}
\end{equation}

Observe that the feedback signals are $x$ in~\eqref{eq: SF} and $\dot{x}$ in~\eqref{eq: SDF}.
Using $K_s$, $N_s$, and $r$ from~\eqref{eq: SF}, the SDF design problem is to 
determine matrices $K_n$ and $N_n$ such that $u_{\mathrm{SF}} = u_{\mathrm{SDF}}$. 
Thus, the state response under~\eqref{eq: SF} is identical 
to that under~\eqref{eq: SDF}. 

{\color{black} We now present a standard result from the SDF literature, stated here without proof. We use this result later to develop the derivative feedback law for the power system application.} Assumptions below are standard\cite{cardim2007design, Kiritsis2023SDF}: 
\begin{enumerate}[label=(\roman*)]
    \item The matrix $A$ is non-singular; that is, $A$ has full rank. 
    \item $\det(A - B K_s) \neq 0$;
    \item The input matrix $B$ has full column rank $m$.
\end{enumerate}

\begin{theorem}(\cite{cardim2007design})\label{thm: SDF}
Consider the system given by~\eqref{eq: LTI} under the control law in~\eqref{eq: SF}. 
Let Assumptions~(i)--(iii) hold. Define
\begin{subequations}\label{eq: SDF_gains}
\begin{align}
    K_n &= K_s (A - B K_s)^{-1}, \label{eq: Kn} \\
    N_n &= (I + K_n B) N_s.      \label{eq: Nn}
\end{align}
\end{subequations}
Then $u_{\mathrm{SF}} = u_{\mathrm{SDF}}$.
\end{theorem}

If Assumption~(i) does not hold, the uncontrolled dynamics $\dot{x}(t)=A x(t)$ would not permit full reconstruction of $x(t)$ from $\dot{x}(t)$, preventing derivative feedback from fully regulating the state. Hence, non-singularity of $A$ is necessary for the existence of SDF control, with sufficiency being established in~\cite{Kiritsis2023SDF}.
Assumption~(ii) ensures that the system in~\eqref{eq: LTI} remains asymptotically stable under $u_{\mathrm{SDF}}$ whenever it is stable under $u_{\mathrm{SF}}$. {\color{black} Assumption (iii) ensures
that the input channels are
independent.}
These assumptions are justified in the context of power systems in Section~\ref{sec: SDF control}.

\section{SDF Control for Oscillation Damping}\label{sec: SDF control}
This section adapts result in Theorem~\ref{thm: SDF} to solve the LFO
damping problem. 
To this end, two aspects must be considered: first, to account for the disturbance $\Delta P$ that appears in \eqref{eq: MMPS} but not in \eqref{eq: LTI}; and second, to justify if the assumptions underlying Theorem~\ref{thm: SDF} hold for the power-system.
We begin with the first aspect and present our key result.  

 
\begin{lemma}\label{lma: SDF_power}
Consider the system $\dot{x} = A x + B (u + \Delta P)$, where $\Delta P$ denotes a disturbance. 
Let $u_{\mathrm{SF}} = -K_s x$ and define the SDF control law as $u_{\mathrm{SDF}} = -K_n \dot{x} + K_n B (\Delta P)$. 
Then, under Assumptions~(i)--(iii), $u_{\mathrm{SF}} = u_{\mathrm{SDF}}$.
\end{lemma}

\begin{proof}
Set $N_s = I_{m \times m}$ in~\eqref{eq: Nn} and choose the reference $r = \Delta P$. Then, the closed-loop dynamics of $\dot{x} = A x + B (u + \Delta P)$ under $u_{\mathrm{SF}}$ and $u_{\mathrm{SDF}}$ defined in the hypothesis are identical to those of the system in~\eqref{eq: LTI} under the control laws~\eqref{eq: SF} and~\eqref{eq: SDF}, respectively. 
The conclusion follows from Theorem~\ref{thm: SDF}.
\end{proof}

We now comment on the utility of Lemma~\ref{lma: SDF_power} for the damping problem. Suppose the complete state $x = (\Delta \delta, \Delta \omega)$ is available. 
Then, using pole-placement or linear-quadratic regulator (LQR) methods, one can design SF control $u_\text{SF} = -K_s x$ to achieve performance equal to or better than a
fixed-architecture controller (such as the one developed in~\cite{neely2013benefits}). For instance, for the two-area system,
by selecting,
\begin{align}
K_s =
\begin{bmatrix}
0 & 0 & k_d & -k_d \\
0 & 0 & -k_d & k_d
\end{bmatrix}, 
\end{align}
we recover the control scheme presented in \eqref{eq: simple_control_two_area_system}. 

However, estimating relative rotor angles is more challenging for the reasons mentioned in Section \ref{Intro}, 
limiting the applicability of $u_{\mathrm{SF}}$. 
Lemma~\ref{lma: SDF_power} shows that equivalent closed-loop performance can be achieved with the SDF control $u_{\mathrm{SDF}}$ by using the state derivatives as feedback. The required state derivatives $\Delta\dot{\delta}=\Delta \omega$ (speed deviations) and $\Delta \dot{\omega}$ (acceleration) can be easily estimated from PMUs using FDF.
Finally, using $u_{\mathrm{SDF}}$ in Lemma~\ref{lma: SDF_power} requires us to know the disturbance $\Delta P$, which {\color{black} can be estimated from prior information about the disturbance type or using an unknown-input observer \cite{abooshahab2019disturbance}}.

The above discussion, together with Lemma~\ref{lma: SDF_power}, suggests a straightforward procedure for designing SDF control.
First, compute a static gain $K_s$ using any standard state-feedback method (we use LQR in simulations), and then use this $K_s$ to obtain the SDF gain matrices given by \eqref{eq: SDF_gains}.
It is worth noting a subtle but important distinction between the structures of $u_{\mathrm{SF}}$ and $u_{\mathrm{SDF}}$: while $u_{\mathrm{SF}}$ is purely feedback-based, $u_{\mathrm{SDF}}$ includes an additional feed-forward term $K_n B \Delta P$.
This does not represent a limitation but rather reflects the slightly higher implementation complexity inherent to SDF control.

We now justify the assumptions of Lemma~\ref{lma: SDF_power}.
Assumption~(ii) is standard for any control design problem, while
Assumption~(iii) holds for the power system control problem since the input matrix $B$ has full column rank owing to the invertibility of the inertia matrix $M$ in~\eqref{eq: MMPS}.
However, Assumption~(i) needs $A$ to be non-singular, which
which might not hold for the power system in~\eqref{eq: MMPS}. The singularity arises from the matrix $T$, which has a Laplacian structure and is singular due to the angle invariance of multi-machine systems. We now show that $A$ is non-singular iff $T$ is non-singular. 

\begin{proposition}\label{prop}
Consider the system matrix $A$ in \eqref{eq: MMPS} with non-singular $M$. Then $A$ is non-singular iff $T$ is non-singular. 
\label{Prop}
\end{proposition}

\begin{proof}
Swapping the two block columns of $A$ yields
\begin{align}
A' =
\begin{bmatrix}
I & 0 \\
-M^{-1} D & -M^{-1} T
\end{bmatrix}_{2N\times 2N}.
\end{align}
The first block column matrix of $A'$ has rank $N$ and is linearly independent of the second block. The second column block has full column rank iff $M^{-1}T$ is full column rank. Since $M$ is square and invertible, $M^{-1}T$ has full column rank iff $T$ does. Thus, $A'$ (and hence $A$) is full column rank.
\end{proof}

We make $T$ non-singular by adding self-stiffness terms to the diagonal making it diagonally dominant, and hence, non-singular \cite{HornJohnson2013}. 
Physically this means that the swing-equation of the $i$-th area has the stiffness term $\tilde{T}_{ii}\Delta\delta_i$. These terms contribute to small shunt/load connections at each area bus, and can be tuned by changing the transient reactance of the area. 
Note that not all $\tilde{T}_{ii}\ne 0$. The simulation section provides details on their selection. We conclude this section with an illustration of the above proposition for the two-area system.

\begin{example}
For the two-area system, the system matrix is
\begin{align}
A =
\begin{bmatrix}
0 & 0 & 1 & 0 \\
0 & 0 & 0 & 1 \\
-\frac{T_{12}+\tilde{T}_{11}}{M_1} & \frac{T_{12}}{M_1} & -\frac{D_1}{M_1} & 0 \\
\frac{T_{12}}{M_2} & -\frac{T_{12}+\tilde{T}_{22}}{M_2} & 0 & -\frac{D_2}{M_2}
\end{bmatrix}.
\end{align}
When $\tilde{T}_{11}=\tilde{T}_{22}=0$, the first two columns are negatives of each other and hence linearly dependent, making $A$ singular regardless of the values of the inertia coefficients or the synchronizing torque $T_{12}$. 
Introducing non-zero $\tilde{T}_{ii}$ terms renders these columns linearly independent, and the earlier proposition then guarantees the nonsingularity of $A$.
\end{example}

\section{Simulation Results and Discussion}
We show the performance of SDF damping control through time-domain simulations of the two- and three-area power systems described in~\cite{neely2013benefits}. For benchmarking, the frequency-difference (FD) control \cite{neely2013benefits} is included for comparison. As discussed in Section~\ref{sec: SDF control}, computing the SDF gain matrix $K_n$ in~\eqref{eq: Kn} requires first obtaining the static feedback gain $K_s$, which was computed using the LQR command in MATLAB with weighting matrices $Q = \mathrm{diag}[10,\,1,\,10,\,1]$ and $R = 2I_m$.

We begin with the two-area system; herein, the area inertia constants \( M_1 = M_2 = 6\,\mathrm{s} \), damping coefficients \( D_1 = D_2 = 1.2 \), and the synchronizing coefficient \( T_{12} = 3.132 \). To ensure full-rank system dynamics (see Proposition \ref{prop}), we add self-stiffness terms \( \tilde{T}_{11} = \varepsilon_1 T_{12} \) and \( \tilde{T}_{22} = \varepsilon_2 T_{21} \), where \( \varepsilon_{1}, \varepsilon_{2} \in [0.03, 0.10] \). Recall that $\Delta\omega\approx \Delta f$ (see Section \ref{sec: prelims}).

\vspace{-1em}

\subsection{Small Load Disturbance}
At $t = 5\mathrm{s}$, the load in A1 is changed
from $0.0$ to $-0.01$~p.u. and restored to its nominal value at $t = 7\mathrm{s}$.
The \textit{identical} SF and SDF control signals, shown in Fig.~\ref{fig:Control}, validate Lemma~\ref{lma: SDF_power}.
From Fig.~\ref{fig:Nominal Disturbance}, we observe that
inter-area oscillation damping (captured by $\Delta f_2 - \Delta f_1$ in the bottom-row) is comparable for both FD and SDF controls, with SDF exhibiting slightly lower peaks. More importantly, for intra-area
quantities $\Delta \delta$ and $\Delta f$, {\color{black} SDF provides superior damping performance, characterized by reduced overshoot and faster settling (see the last column in Table \ref{tab_quantitiave})}. These benefits are achieved with only a modest increase in control effort.

\begin{figure}
    \centering
    \includegraphics[width=0.8\linewidth]{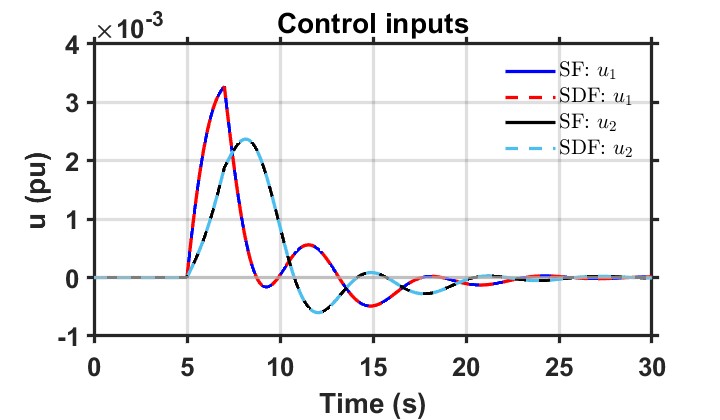}
    \vspace{-0.5em}
    \caption{\small Control inputs for the two-area system under SF and SDF.}
    \label{fig:Control}
    \vspace{-1em}
\end{figure}
\begin{figure}
    \centering
    \includegraphics[width=1\linewidth]{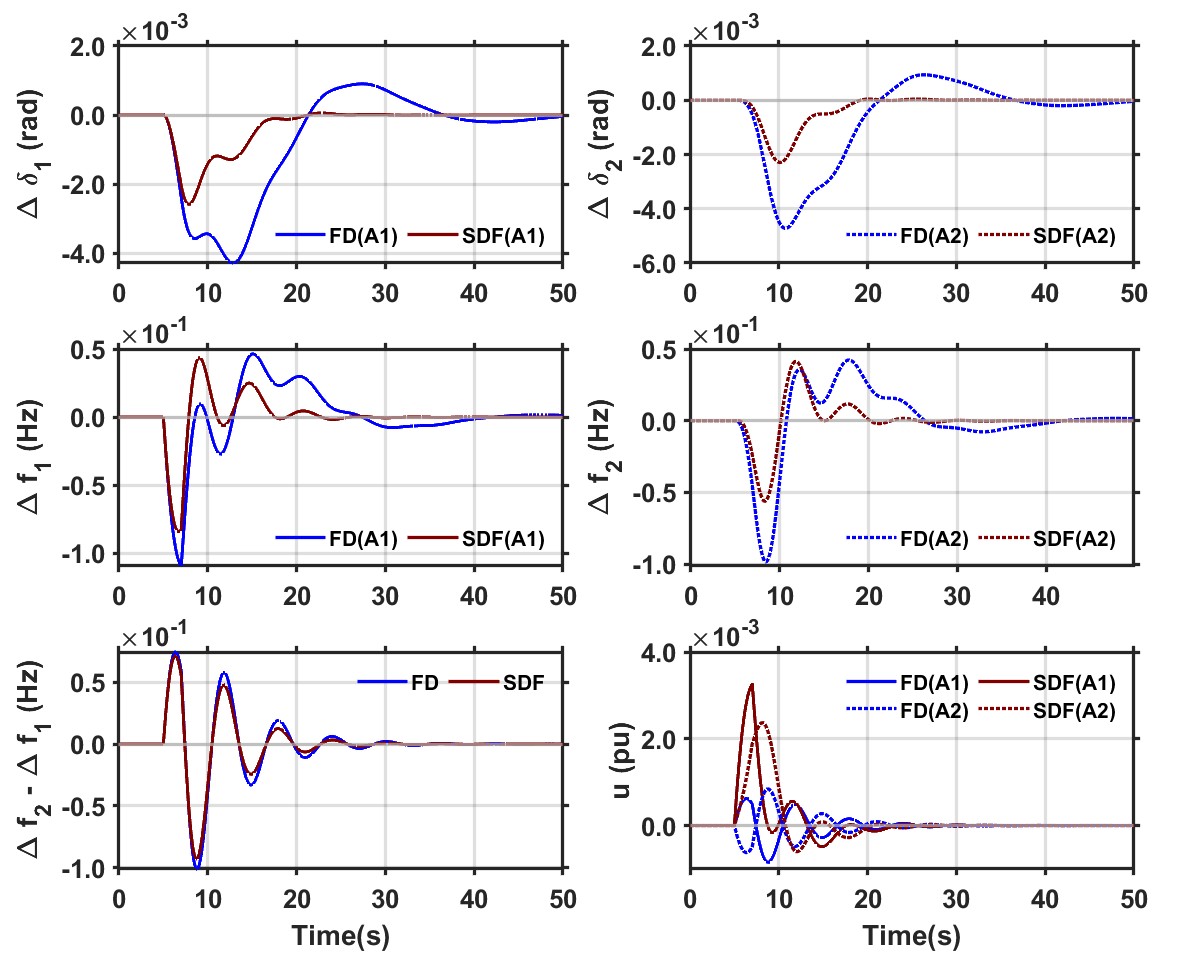}
    \vspace{-1.5em}
    \caption{\small System responses for \(1\%\) load increase in Area 1. In the legend, A$\mathrm{j}$ refers to the $\mathrm{j}$-th area.}
    \label{fig:Nominal Disturbance}
\end{figure} 

\begin{table}[ht!]
\centering
\caption{\small Maximum peak deviation and transient time comparison of FD and SDF controllers}
\begin{tabular}{lccccc}
\toprule
\multirow{2}{*}{\textbf{Signal}} &
\multicolumn{2}{c}{\textbf{Peak Deviation}} &
\multicolumn{3}{c}{\textbf{Transient Time (s)}} \\
\cmidrule(lr){2-3} \cmidrule(lr){4-6}
 & FD & SDF & FD & SDF & Improvement (\%) \\
\midrule
$\Delta\delta_1$ (rad) & 0.0043 & 0.0026 & 48.06 & 23.53 & 51.04 \\
$\Delta\delta_2$ (rad) & 0.0047 & 0.0023 & 47.76 & 18.76 & 60.72 \\
$\Delta f_1$ (Hz)      & 0.1093 & 0.0840 & 39.04 & 24.27 & 37.83 \\
$\Delta f_2$ (Hz)      & 0.0983 & 0.0561 & 39.69 & 24.70 & 37.77 \\
$\Delta f_{21}$ (Hz)   & 0.1011 & 0.0929 & 30.15 & 25.02 & 17.02 \\
\bottomrule
\end{tabular}\label{tab_quantitiave}
\end{table}

To assess robustness of the SDF control under measurement uncertainty, we introduce additive Gaussian noise to the states and their derivatives appropriately. The noise bounds follow PMU steady-state accuracy requirements at a nominal frequency of \(f_{\mathrm{nom}} = 60\,\mathrm{Hz}\), with \(3\sigma_f \leq 0.005\,\mathrm{Hz}\) and \(3\sigma_\delta \leq 0.573^{\circ}\)\cite{IEEEStandards}.  
The system responses in Fig.~\ref{fig:Measurement Error} show trends similar to the noise-free case, though the control inputs exhibit slight jitter, warranting further study using techniques such as low-pass filtering.
\begin{figure}
    \centering
    \includegraphics[width=1\linewidth]{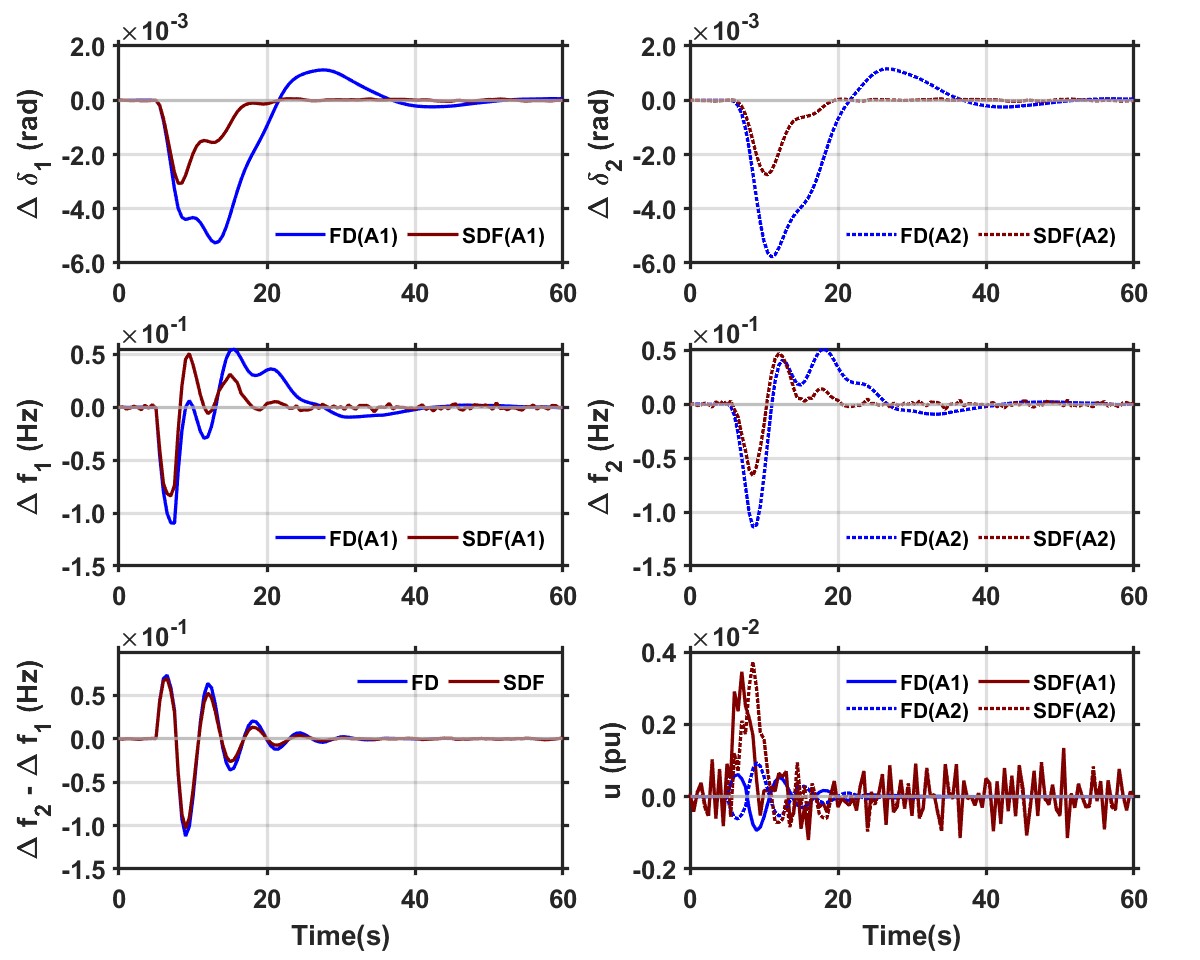}
    \vspace{-1.75em}
    \caption{\small System responses under measurement noise consistent with PMU accuracy limits.}
    \label{fig:Measurement Error}
    \vspace{-1em}
\end{figure}
\setlength{\textfloatsep}{5pt plus 2pt minus 2pt} 

The performance of the SDF is further evaluated under more realistic disturbances, including high-intensity fault currents and periodic large load fluctuations.

\begin{figure}
    \centering
    \includegraphics[width=1\linewidth]{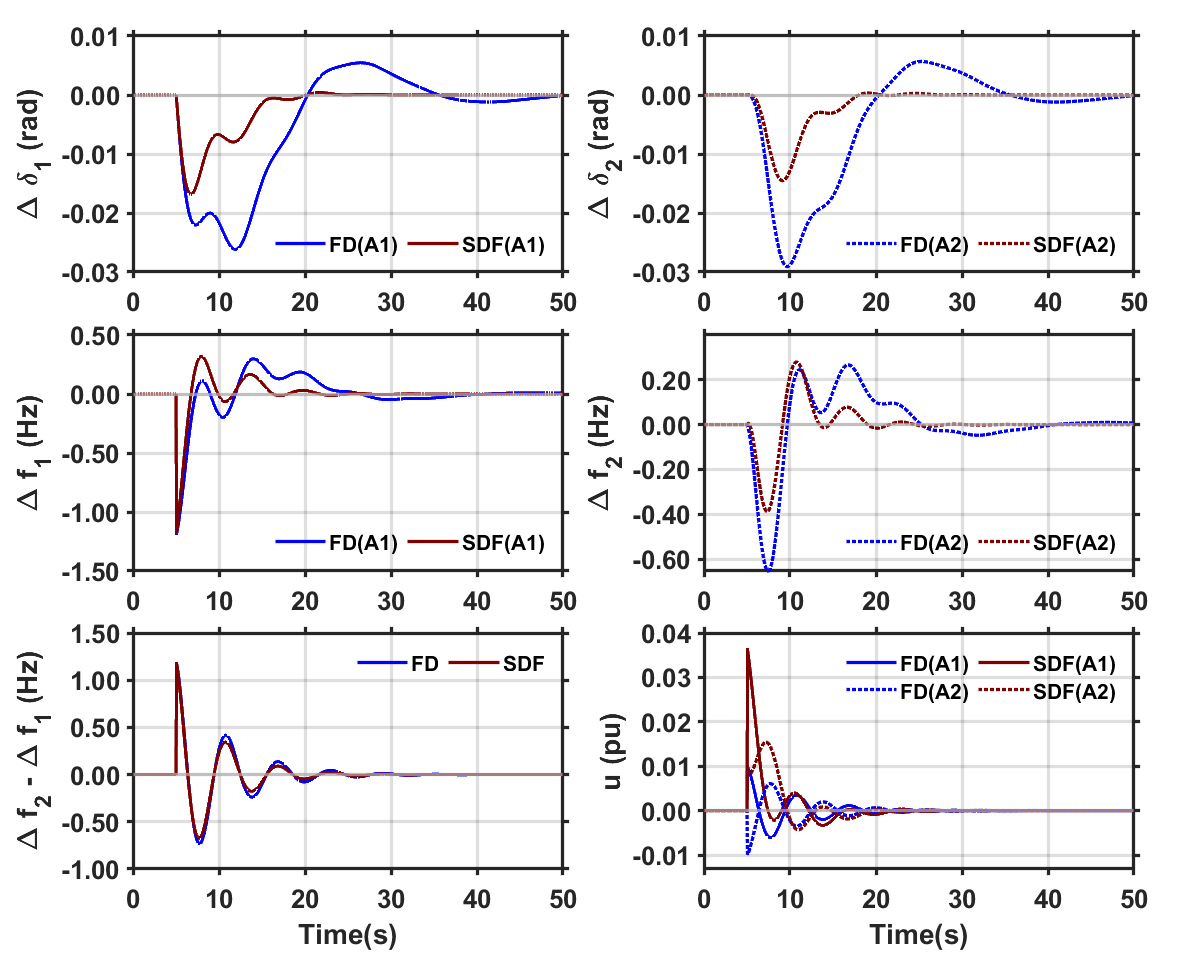}
    \vspace{-1.75em}
    \caption{\small System responses to a short-duration fault in Area 1}
    \label{fig:Fault}
\end{figure}

\begin{figure*}[t]
    \centering
    \includegraphics[width=1\linewidth]{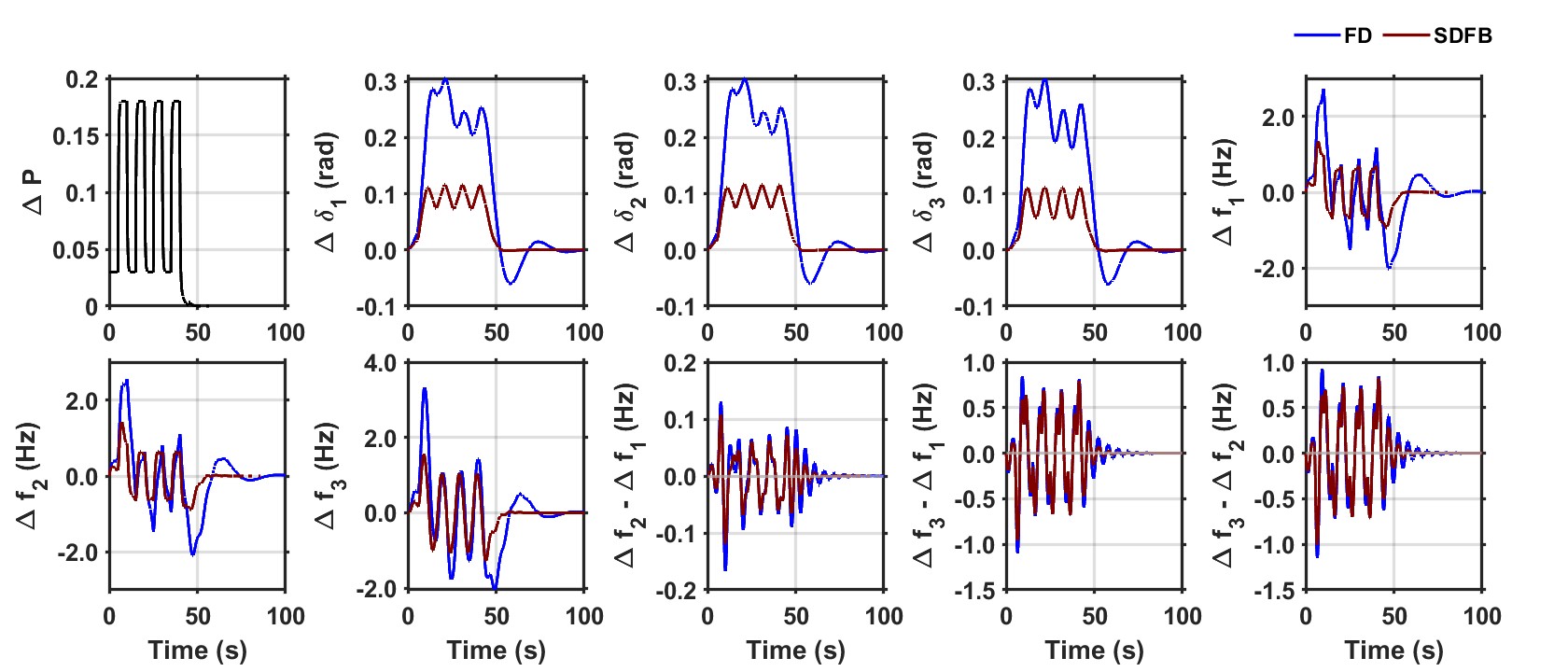}
    \vspace{-1.75em}
    \caption{\small System response under AI data-center load fluctuation.}
    \label{fig:Data Center Load}
    \vspace{-1.25em}
\end{figure*}
\setlength{\textfloatsep}{5pt plus 2pt minus 2pt} 

\subsection{Fault Disturbance} 
We study a fault disturbance with an amplitude two times larger than the nominal ambient variation. This disturbance is applied at \(t = 5\,\mathrm{s}\) and cleared at \(t = 5.1\,\mathrm{s}\). The time-domain responses in Fig.~\ref{fig:Fault} again show trends similar to Fig.~\ref{fig:Nominal Disturbance}; however, the transient control effort of SDF is higher than that of FD, which is expected and reasonable given the improved system responses. 

\subsection{Data Center Load Fluctuation}
We consider load fluctuation of a data center in A1 of the three-area system \cite{neely2013benefits}. The fluctuation is characterized by burst-type load change between $0.03$~p.u. and $0.18$~p.u. occurring in successive cycles over the interval \(t = 0\)–\(40\,\mathrm{s}\). Even for this extreme scenario,
our SDF control achieved better system response
(see Fig.~\ref{fig:Data Center Load}).

\section{Conclusion}
We examined the effectiveness of SDF control in improving the damping of inter-area and intra-area oscillations (in frequency and angle differences) in multi-area power systems. The SDF control is easy to implement: an LQR-based optimal feedback gain is first computed, followed by an algebraic transformation to obtain the SDF gain matrix. Numerical studies demonstrate that SDF outperforms the frequency-difference (FD) control scheme, underscoring its practical value using readily available measurements such as frequency and RoCoF. 

Future work will be on: (i) developing SDF for the second-order swing equations rather than their first-order representations in~\eqref{eq: MMPS}, (ii) handling low-inertia systems as in~\cite{anguluri2022parameter}; and (iii) evaluating the SDF control on larger-scale power system models such as the 240-bus reduced WECC system \cite{yuan2020developing}.

\bibliographystyle{IEEEtran}
\bibliography{ref.bib}

@inproceedings{neely2013benefits,
  title={The Benefits of Energy Storage Combined with {HVDC} Transmission Power Modulation for Mitigating Inter-Area Oscillations},
  author={Neely, J and Elliott, R and Byrne, R and Schoenwald, D and Trudnowski, D},
  booktitle={Electrical Energy Storage Applications and Technologies Conf., San Diego, CA},
  year={2013}
}

@inproceedings{abooshahab2019disturbance,
  title={Disturbance and state estimation in partially known power networks},
  author={Abooshahab, Mohammad Ali and Hovd, Morten and Bitmead, Robert R},
  booktitle={IEEE Conference on Control Technology and Applications (CCTA)},
  pages={98--105},
  year={2019},
  organization={IEEE}
}

@article{liu2020model,
  title={Model-independent derivative control delay compensation methods for power systems},
  author={Liu, Muyang and Dassios, Ioannis and Tzounas, Georgios and Milano, Federico},
  journal={Energies},
  volume={13},
  number={2},
  pages={342},
  year={2020},
  publisher={MDPI}
}

@inproceedings{anguluri2022parameter,
  title={Parameter Estimation in Ill-conditioned Low-inertia Power Systems},
  author={Anguluri, Rajasekhar and Sankar, Lalitha and Kosut, Oliver},
  booktitle={2022 North American Power Symposium (NAPS)},
  pages={1--6},
  year={2022},
  organization={IEEE}
}

@article{rakhshani2016inertia,
  title={Inertia emulation in {AC/DC} interconnected power systems using derivative technique considering frequency measurement effects},
  author={Rakhshani, Elyas and Rodriguez, Pedro},
  journal={IEEE Transactions on Power Systems},
  volume={32},
  number={5},
  pages={3338--3351},
  year={2016},
  publisher={IEEE}
}

@article{Kiritsis2023SDF,
  author    = {Konstadinos H. Kiritsis},
  title     = {Stabilization of Linear Time-Invariant Systems by State-Derivative Feedback},
  journal   = {WSEAS Transactions on Systems and Control},
  volume    = {18},
  pages     = {65--72},
  year      = {2023}
}

@book{HornJohnson2013,
  author    = {Roger A. Horn and Charles R. Johnson},
  title     = {Matrix Analysis},
  edition   = {2nd},
  publisher = {Cambridge University Press},
  address   = {Cambridge, UK},
  year      = {2013},
  isbn      = {978-0-521-54823-6}
}

@article{cardim2007design,
  title={Design of state-derivative feedback controllers using a state feedback control design},
  author={Cardim, Rodrigo and Teixeira, Marcelo CM and Assun{\c{c}}{\~A}o, Edvaldo and Covacic, M{\'a}rcio R},
  journal={IFAC Proceedings Volumes},
  volume={40},
  number={20},
  pages={22--27},
  year={2007},
  publisher={Elsevier}
}

@misc{IEEEStandards,
  author={},
  journal={IEC/IEEE 60255-118-1:2018}, 
  title={{IEEE/IEC} International Standard - Measuring relays and protection equipment - Part 118-1: Synchrophasor for power systems - Measurements}, 
  year={2018},
  volume={},
  number={},
  pages={1-78},
  doi={10.1109/IEEESTD.2018.8577045}}

@article{pal2013applying,
  title={Applying a robust control technique to damp low frequency oscillations in the {WECC}},
  author={Pal, Anamitra and Thorp, James S and Veda, Santosh S and Centeno, VA},
  journal={International Journal of Electrical Power \& Energy Systems},
  volume={44},
  number={1},
  pages={638--645},
  year={2013},
  publisher={Elsevier}
}

@ARTICLE{pierre2019design,
  author={Pierre, Brian J. and Wilches-Bernal, Felipe and Schoenwald, David A. and Elliott, Ryan T. and Trudnowski, Daniel J. and Byrne, Raymond H. and Neely, Jason C.},
  journal={IEEE Transactions on Power Systems}, 
  title={Design of the Pacific {DC} Intertie Wide Area Damping Controller}, 
  year={2019},
  volume={34},
  number={5},
  pages={3594-3604},
  doi={10.1109/TPWRS.2019.2903782}}

@article{milano2017frequency,
  title={Frequency Divider},
  author={Milano, Federico and Ortega, Alvaro},
  journal={IEEE Transactions on Power Systems},
  volume={32},
  number={2},
  pages={1493--1501},
  year={2017},
  publisher={IEEE}
}

@INPROCEEDINGS{sharma2025practical,
  author={Sharma, Anushka and Pal, Anamitra and Anguluri, Rajasekhar and Chakraborty, Tamojit},
  booktitle={2025 IEEE Power \& Energy Society General Meeting (PESGM)}, 
  title={A Practical Approach Towards Inertia Estimation Using Ambient Synchrophasor Data}, 
  year={2025},
  volume={},
  number={},
  pages={1-5},
  doi={10.1109/PESGM52009.2025.11225304}}

@ARTICLE{chai2025robust,
  author={Chai, Bo and Chan, Shing Chow and Hou, Yunhe},
  journal={IEEE Transactions on Power Systems}, 
  title={Robust Adaptive Fading Unscented {Kalman} Filter for Decentralized Dynamic State Estimation in Power Systems Under Unknown Inputs and Covariance Mismatches}, 
  year={2025},
  volume={40},
  number={6},
  pages={4732-4745},
  doi={10.1109/TPWRS.2025.3570748}}

@article{xu2021direct,
  title={Direct damping feedback control using power electronics-interfaced resources},
  author={Xu, Xin and Sun, Kai},
  journal={IEEE Transactions on Power Systems},
  volume={37},
  number={2},
  pages={1113--1125},
  year={2021},
  publisher={IEEE}
}

@article{abdelaziz2005state,
  title={State derivative feedback by {LQR} for linear time-invariant systems},
  author={Abdelaziz, Taha HS and Val{\'a}{\v{s}}ek, Michael},
  journal={IFAC Proceedings Volumes},
  volume={38},
  number={1},
  pages={435--440},
  year={2005},
  publisher={Elsevier}
}

@article{da2012less,
  title={Less Conservative Control Design for Linear Systems with Polytopic Uncertainties via State-Derivative Feedback},
  author={Da Silva, Emerson RP and Assun{\c{c}}ao, Edvaldo and Teixeira, Marcelo CM and Buzachero, Luiz Francisco S},
  journal={Mathematical Problems in Engineering},
  volume={2012},
  number={1},
  pages={315049},
  year={2012},
  publisher={Wiley Online Library}
}

@inproceedings{yuan2020developing,
  title={Developing a reduced 240-bus {WECC} dynamic model for frequency response study of high renewable integration},
  author={Yuan, Haoyu and Biswas, Reetam Sen and Tan, Jin and Zhang, Yingchen},
  booktitle={2020 IEEE/PES Transmission and Distribution Conference and Exposition (T\&D)},
  pages={1--5},
  year={2020},
  organization={IEEE}
}

\end{document}